\documentclass[12pt]{article} 
\usepackage{natbib,bm,algorithm2e,graphicx,subcaption,caption,placeins,multirow,booktabs,enumerate,fullpage,amsmath,amssymb,amsthm}

\newtheoremstyle{definition}{12pt}{12pt}{\rm}{}{\sffamily}{ }{ }{}
\theoremstyle{definition}
\newtheorem{definition}{\sc Definition}

\newtheoremstyle{corollary}{12pt}{12pt}{\rm}{}{\sffamily}{ }{ }{}
\theoremstyle{corollary}

\newtheoremstyle{remark}{12pt}{12pt}{\rm}{}{\sffamily}{ }{ }{}
\theoremstyle{remark}
\newtheorem{remark}{\sc Remark}

\newtheoremstyle{lemma}{12pt}{12pt}{\rm}{}{\sffamily}{ }{ }{}
\theoremstyle{lemma}
\newtheorem{lemma}{\sc Lemma}

\newtheoremstyle{Assumption}{12pt}{12pt}{\rm}{}{\sffamily}{ }{ }{}
\theoremstyle{Assumption}
\newtheorem{Assumption}{\sc Assumption}

\newtheoremstyle{theorem}{12pt}{12pt}{\rm}{}{\sffamily}{ }{ }{}
\theoremstyle{theorem}
\newtheorem{theorem}{\sc Theorem}

\newtheoremstyle{example}{12pt}{12pt}{\rm}{}{\sffamily}{ }{ }{}
\theoremstyle{example}
\newtheorem{example}{\sc Example}

\newtheorem*{theorem*}{\sc Theorem}

\def\ci{\perp\!\!\!\perp}
\def\nci{\not\perp\!\!\!\perp}

\begin{document}

\title{Estimating causal structure using conditional DAG models}

\author{CHRIS. J. OATES$^\ast$, JIM Q. SMITH\\[4pt]
\textit{Department of Statistics, University of Warwick}, \\ \textit{Coventry, CV4 7AL UK}
\\[2pt]
{c.oates@warwick.ac.uk, j.q.smith@warwick.ac.uk} \\[8pt]
SACH MUKHERJEE\\[4pt]
\textit{MRC Biostatistics Unit and School of Clinical Medicine}, \\ \textit{University of Cambridge, Cambridge, CB2 0SR UK}
\\[2pt]
{sach@mrc-bsu.cam.ac.uk}}

\markboth%
{C. J. Oates, J. Q. Smith and S. Mukherjee}
{Causal inference using conditional DAG models}

\maketitle

\begin{abstract}
{This paper considers  inference of causal structure in a class of graphical models called ``conditional DAGs''.
These are directed acyclic graph (DAG) models with two kinds of variables, primary and secondary. The secondary variables are used to aid in estimation of causal relationships between the primary variables. 
We give causal semantics for this model class and prove that, under certain assumptions, the direction of causal influence  is  identifiable from the joint observational distribution of the primary and secondary variables.
A score-based  approach is developed for estimation of causal structure using these models and consistency results are established.
Empirical results demonstrate  gains compared with formulations  that treat all variables on an equal footing, or that  ignore secondary variables.
The methodology is motivated by applications in molecular biology and is illustrated here using simulated data and in an analysis of proteomic data from the Cancer Genome Atlas.
}
{graphical models, causal inference, directed acyclic graphs, instrumental variables}
\end{abstract}

\section{Introduction}

This paper considers estimation of causal structure among a set of ``primary" variables $(Y_i)_{i \in V}$, using additional ``secondary" variables $(X_i)_{i \in W}$ to aid in estimation. 
The primary variables are those of direct scientific interest while the secondary variables are variables that are known to influence the primary variables, but whose mutual relationships are not of immediate interest and possibly not amenable to inference using the available data.
As we discuss further below, the primary/secondary distinction is common in biostatistical applications and is often dealt with in an {\it ad hoc} manner, for example by leaving some relationships or edges implicit in causal diagrams. Our aim is to define a class of graphical models for this setting and to clarify the conditions under which secondary variables can aid in causal inference. We focus on structural inference in the sense of estimation of the presence or absence of edges in the causal graph rather than estimation of quantitative causal effects.

The fact that primary variables of direct interest are often part of a wider context, including additional secondary variables, presents challenges for graphical modelling and causal inference, since in general the secondary variables will not be independent and simply marginalising may introduce spurious dependencies \citep{Evans2}.
Motivated by this observation, we define ``conditional'' DAG (CDAG) models and discuss their semantics.
Nodes in a CDAG are of two kinds corresponding to primary and secondary variables, and as  detailed below the semantics of CDAGs allow causal inferences to be made about the primary variables $(Y_i)_{i \in V}$ whilst accounting for the secondary variables $(X_i)_{i \in W}$.
To limit scope, we focus on the  setting where each primary variable has a known cause among the secondary variables, specifically we suppose there is a  bijection $\phi : V \rightarrow W$, between the primary and secondary index sets $V$ and $W$, such that for each $i \in V$ a direct causal dependency $X_i \rightarrow Y_{\phi(i)}$ exists. (Throughout, we use the term ``direct" in the sense of \citet{Pearl} and note that the causal influence need not be {\it physically} direct, but rather may permit non-confounding intermediate variables).
Under  explicit assumptions we show that such secondary variables can aid in causal inference for the primary variables, because known causal relationships between secondary and primary variables render ``primary-to-primary" causal links of the form $Y_i \rightarrow Y_j$ identifiable from joint data on primary and secondary variables.
We put forward score-based estimators of CDAG structure that we show are asymptotically consistent under certain conditions; importantly, independence assumptions on the secondary variables are not needed. 

This work was motivated by current efforts in molecular biology aimed at exploiting high-throughput  biomolecular data to better understand causal molecular mechanisms, such as those involved in gene regulation or protein signaling.
A notable feature of molecular biology is the fact that some causal  links are relatively clearly defined by known sequence specificity. 
For example DNA sequence variation has a causal influence on the level of corresponding mRNA; mRNAs have a causal influence on corresponding total protein levels; and total protein levels have a causal influence on levels of post-translationally modified protein.
This means that in a study involving a certain molecular variable (a protein, say), a subset of the causal influences upon it may be clear at the outset (e.g. the corresponding mRNA) and typically it is the unknown influences that are the subject of the study. Then, it is natural to ask whether accounting for the known influences can aid in identification of the unknown influences. For example,  if interest focuses on causal relationships between proteins, known mRNA-protein links could be exploited to aid in causal identification at the protein-protein level. 
 We show below an  example with certain (post-translationally modified) proteins as the primary variables and total protein levels as secondary variables.

Our development of the CDAG can be considered dual to the acyclic directed mixed graphs (ADMGs) developed by \cite{Evans2}, in the sense that we investigate conditioning as an alternative to marginalisation.
In this respect our work mirrors recently developed undirected graphical models called conditional graphical models \citep[CGMs; ][]{Li,Cai}
In CGMs, Gaussian random variables $(Y_k)_{k \in V}$ satisfy 
\begin{eqnarray}
Y_i \ci Y_j | (Y_k)_{k \in V \setminus \{i,j\}}, (X_k)_{k \in W} \text{ if and only if } (i,j) \notin G
\end{eqnarray}
where $G$ is an undirected acyclic graph and $(X_k)_{k \in W}$ are auxiliary random variables that are conditioned upon.
CGMs have recently been applied to gene expression data $(Y_i)_{i \in V}$ with the $(X_i)_{i \in W}$ corresponding to single nucleotide polymorphisms (SNPs) \citep{Zhang2} and with the secondary variables $(X_i)_{i \in W}$ corresponding to expression qualitative trait loci (e-QTL) data \citep{Logsdon,Yin,Cai2}, the latter being recently extended to jointly estimate several such graphical models in \cite{Chun}.
Also in the context of undirected graphs, \cite{vanW} recently considered encoding a bijection between DNA copy number and mRNA expression levels into inference.
Our work  complements these efforts using directed models that are arguably  more appropriate for  causal inference \citep{Lauritzen2}. CDAGs are also related to instrumental variables and Mendelian randomisation approaches \citep{Didilez} that we discuss below (Section \ref{section_causal_CDAG}).

The class of CDAGs shares some  similarity with the influence diagrams (IDs) introduced by \cite{Dawid2} as an extension of DAGs that distinguish between variable nodes and decision nodes.
This generalised the augmented DAGs of \cite{Spirtes,Lauritzen,Pearl} in which each variable node is associated with a decision node that represents an intervention on the corresponding variable.
However, the semantics of IDs are not well suited to the scientific contexts that we consider, where  secondary nodes represent variables to be observed, not the outcomes of decisions.
The notion of a non-atomic intervention \citep{Pearl3}, where many variables are intervened upon simultaneously, shares similarity with CDAGs in the sense that the secondary variables are in general non-independent.
However again the semantics differ, since our secondary nodes represent random variables rather than interventions. 
In a different direction, \cite{Neto} recently observed that the use of e-QTL data $(X_i)_{i \in W}$ can help to identify causal relationships among gene expression levels $(Y_i)_{i \in V}$.
However, \cite{Neto} require independence of the  $(X_i)_{i \in W}$;
this is too restrictive for general  settings, including in molecular biology, since the secondary variables will typically themselves be subject to regulation and far from independent.

This paper begins in Sec. \ref{methods} by defining CDAGs and discussing  identifiability of their structure from observational data on primary and secondary variables. 
Sufficient conditions are then given for  consistent estimation of CDAG structure along with an algorithm based on integer linear programming. The methodology is illustrated in Section \ref{results} on simulated data, including datasets that violate CDAG assumptions, and on proteomic data from cancer patient samples, the latter from the Cancer Genome Atlas (TCGA) ``pan-cancer" study.

\begin{figure}[t]
\centering
\begin{subfigure}[]{0.45\textwidth}
\centering
\includegraphics[width = \textwidth,clip,trim = 3cm 24cm 9cm 2.5cm]{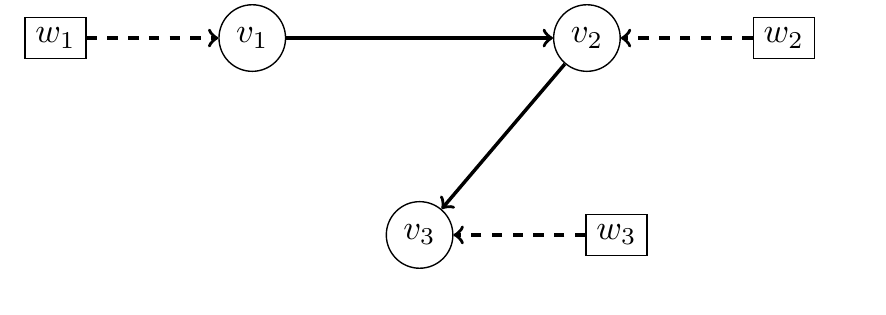}
\caption{}
\label{illustrate}
\end{subfigure} 
\begin{subfigure}[c]{0.45\textwidth}
\centering
\vspace{10pt}
\begin{tabular}{|l|c|c|c|} \hline
type & index & node & variable \\ \hline
primary & $i \in V$ & $v_i \in N(V)$ & $Y_i$ \\
secondary & $i \in W$ & $w_i \in N(W)$ & $X_i$ \\ \hline
\end{tabular}
\vspace{20pt}
\caption{}
\label{illustrate2}
\end{subfigure} 
\caption{A conditional DAG model with primary nodes $N(V) = \{v_1,v_2,v_3\}$ and secondary nodes $N(W) = \{w_1,w_2,w_3\}$. Here primary nodes represent primary random variables $(Y_i)_{i \in V}$ and solid arrows correspond to a DAG $G$ on these vertices. 
Square nodes are secondary variables $(X_i)_{i \in W}$ that, in the causal interpretation of CDAGs, represent known direct causes of the corresponding $(Y_i)_{i \in V}$ (dashed arrows represent known relationships; the random variables $(X_i)_{i \in W}$ need not be independent).
The name ``conditional" DAG refers to the fact that conditional upon $(X_i)_{i \in W}$, the solid arrows encode conditional independence relationships among the $(Y_i)_{i \in V}$.}
\end{figure}

\section{Methodology} \label{methods}

\subsection{A statistical framework for conditional DAG models} \label{defs}

Consider index sets $V$, $W$ and a bijection $\phi:V \rightarrow W$ between them.
We will distinguish between the nodes in graphs and the random variables (RVs) that they represent. 
Specifically, indices correspond to nodes in graphical models; this is signified by the notation $N(V) = \{v_1,\dots,v_p\}$ and $N(W) = \{w_1,\dots,w_p\}$.
Each node $v_i \in N(V)$ corresponds to a primary RV $Y_i$ and similarly each node $w_i \in N(W)$ corresponds to a secondary RV $X_i$.

\begin{definition}[CDAG]
A \emph{conditional DAG} (CDAG) $\overline{G}$, with primary and secondary index sets $V$, $W$ respectively and a bijection $\phi$ between them, is a DAG on the primary node set $N(V)$ with additional  directed edges from each secondary node $w_i \in N(W)$ to its corresponding primary node $v_{\phi(i)} \in N(V)$. 
\end{definition}

In other words, a CDAG $\overline{G}$ has  node set $N(V) \cup N(W)$ and an edge set that can be generated by starting with a DAG on the primary nodes $N(V)$ and  adding a directed edge from each secondary node in $N(W)$ to its corresponding primary node in $N(V)$, with the correspondence specified by the bijection $\phi$. 
An example of a CDAG is shown in Fig. \ref{illustrate}.
To further distinguish $V$ and $W$ in the graphical representation we employ circular and square vertices respectively.
In addition we use dashed lines to represent edges that are required by definition and must therefore be present in any CDAG $\overline{G}$.

Since the DAG on the primary nodes $N(V)$ is of particular interest, throughout we use $G$ to denote a DAG on $N(V)$. We use  $\mathcal{G}$ to denote the set of all possible DAGs with $|V|$ vertices.
For notational clarity, and without loss of generality, below we take the bijection to simply be the identity map $\phi(i) = i$.
The parents of node $v_i$ in a DAG $G$
are indexed by  $\text{pa}_G(i) \subseteq V \setminus \{i\}$.
Write $\text{an}_{\overline{G}}(S)$ for the ancestors of nodes $S \subseteq N(V) \cup N(W)$ in the CDAG $\overline{G}$ (which by definition includes the nodes in $S$).
For disjoint sets of nodes $A,B,C$ in an undirected graph, we say that $C$ \emph{separates} $A$ and $B$ if every path between a node in $A$ and a node in $B$ in the graph contains a node in $C$.

\begin{definition}[$c$-separation]
Consider disjoint $A,B,C \subseteq N(V) \cup N(W)$ and a CDAG $\overline{G}$.
We say that $A$ and $B$ are \emph{$c$-separated} by $C$ in $\overline{G}$, written $A \ci B | C \; [\overline{G}]$, when $C$ separates $A$ and $B$ in the undirected graph $U_4$ that is formed as follows: (i) Take the subgraph $U_1$ induced by $\text{an}_{\overline{G}}(A \cup B \cup C)$. (ii) Moralise $U_1$ obtain $U_2$ (i.e. join with an undirected edge any parents of a common child that are not already connected by a directed edge). (iii) Take the skeleton of the moralised subgraph $U_2$ to obtain $U_3$ (i.e. remove the arrowheads). (iv) Add an undirected edge between every pair of nodes in $N(W)$ to obtain $U_4$.
\end{definition}

The $c$-separation procedure is illustrated in Fig. \ref{csep}, where we show that $v_3$ is not $c$-separated from $v_1$ by the set $\{v_2\}$ in the CDAG from Fig. \ref{illustrate}.

\begin{remark}
The classical notion of $d$-separation for DAGs is equivalent to omitting step (iv) in $c$-separation. 
Notice that $v_3$ \emph{is} $d$-separated from $v_1$ by the set $\{v_2\}$ in the DAG $G$, so that we really do require this custom notion of separation for CDAGs.
\end{remark}

\begin{figure}[t]
\centering
\begin{subfigure}[t]{0.45\textwidth}
\centering
\includegraphics[width = \textwidth,clip,trim = 3cm 24.22cm 9cm 2.5cm]{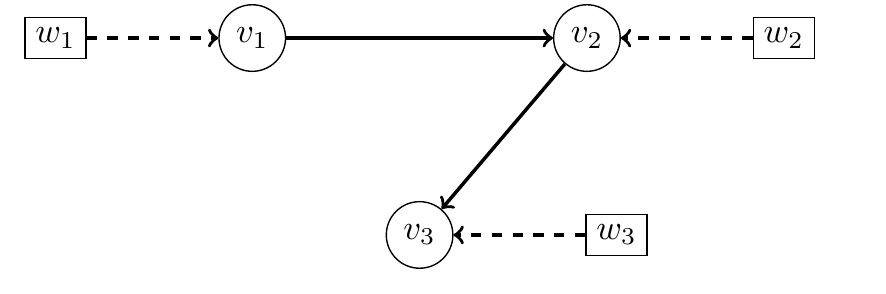}
\caption{Step (i)}
\label{U1}
\end{subfigure} 
\begin{subfigure}[t]{0.45\textwidth}
\centering
\includegraphics[width = \textwidth,clip,trim = 3cm 23.5cm 9cm 2.5cm]{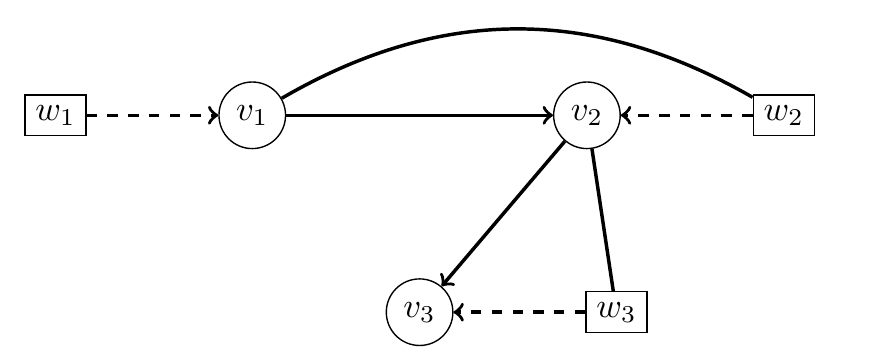}
\caption{Step (ii)}
\label{U2}
\end{subfigure} 

\begin{subfigure}[t]{0.45\textwidth}
\centering
\includegraphics[width = \textwidth,clip,trim = 3cm 23.2cm 9cm 2.5cm]{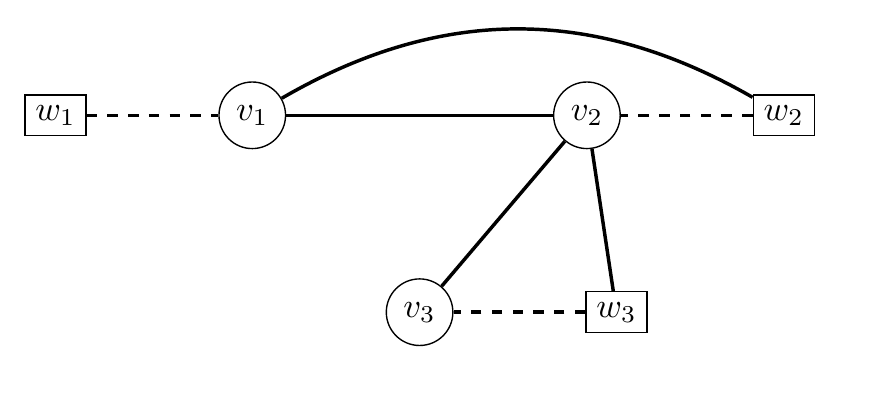}
\caption{Step (iii)}
\label{U3}
\end{subfigure} 
\begin{subfigure}[t]{0.45\textwidth}
\centering
\includegraphics[width = \textwidth,clip,trim = 3cm 22.85cm 9cm 2.5cm]{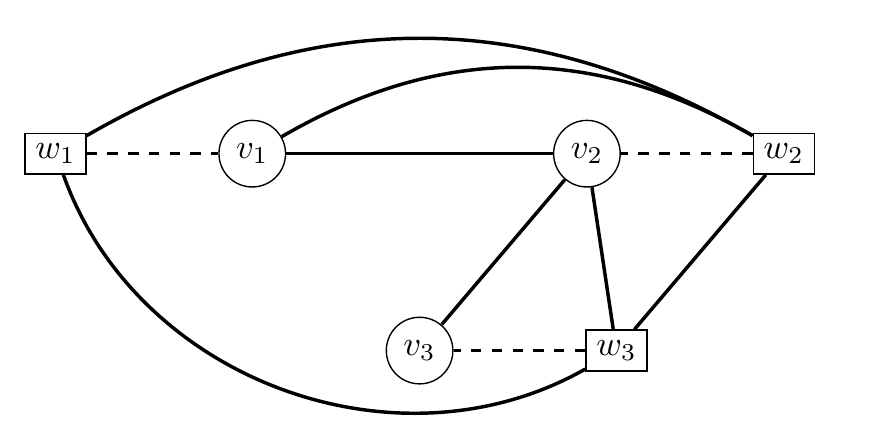}
\caption{Step (iv)}
\label{U4}
\end{subfigure} 
\caption{Illustrating $c$-separation. Here we ask whether $v_3$ and $v_1$ are $c$-separated by $\{v_2\}$ in $\overline{G}$, the CDAG shown in Fig. \ref{illustrate}.
[Step (i): Take the subgraph induced by $\text{an}_{\overline{G}}(\{v_1,v_2,v_3\})$. 
Step (ii): Moralise this subgraph (i.e. join with an undirected edge any parents of a common child that are not already connected by a directed edge).
Step (iii): Take the skeleton of the moralised subgraph (i.e. remove the arrowheads). 
Step (iv): Add an undirected edge between every pair $(w_i,w_j)$.
In the final panel (d) we ask whether there exists a path from $v_3$ to $v_1$ that does not pass through $v_2$; the answer is positive (e.g. $v_3 - w_3 - w_1 - v_1$) and hence we conclude that $v_3 \nci v_1 | v_2 [\overline{G}]$, i.e. $v_3$ and $v_1$ are not $c$-separated by $\{v_2\}$ in $\overline{G}$.]}
\label{csep}
\end{figure}

The topology (i.e. the set of edges) of a CDAG carries formal (potentially causal) semantics on the primary RVs, conditional on the secondary RVs, as specified below.
Write $T(S)$ for the collection of triples $\langle A,B|C \rangle$ where $A,B,C$ are disjoint subsets of $S$.

\begin{definition}[Independence model] \label{indep model}
The CDAG $\overline{G}$, together with $c$-separation, implies a formal \emph{independence model} \citep[p.12 of][]{Studeny}
\begin{eqnarray}
\mathcal{M}_{G} = \{ \langle A,B|C \rangle \in T(N(V) \cup N(W)) : A \ci B | C \; [\overline{G}] \}
\end{eqnarray}
where $\langle A,B|C \rangle \in \mathcal{M}_G$ carries the  interpretation that the RVs corresponding to $A$ are conditionally independent of the RVs corresponding to $B$ when given the RVs corresponding to $C$.
We will write $A \ci B | C \; [\mathcal{M}_G]$ as a shorthand for $\langle A,B|C \rangle \in \mathcal{M}_G$.
\end{definition}

\begin{remark}
An independence model $\mathcal{M}_G$ does not contain any information on the structure of the marginal distribution $\mathbb{P}^{(X_i)}$ of the secondary variables, due to the additional step (iv) in $c$-separation.
\end{remark}

\begin{lemma}[Equivalence classes] \label{equiv}
The map $G \mapsto \mathcal{M}_G$ is an injection.
\end{lemma}
\begin{proof}
Consider two distinct DAGs $G,H \in \mathcal{G}$ and suppose that, without loss of generality, the edge $v_i \rightarrow v_j$ belongs to $G$ and not to $H$.
It suffices to show that $\mathcal{M}_G \neq \mathcal{M}_H$.
First notice that $G$ represents the relations (i) $w_i \nci v_j|w_j \; [\overline{G}]$, (ii) $w_i \nci v_j | w_j, (v_k)_{k \in V \setminus \{i,j\}} \; [\overline{G}]$, and (iii) $w_j \ci v_i | w_i \; [\overline{G}]$.
(These can each be directly  verified by $c$-separation.)
We show below that $H$ cannot also represent (i-iii) and hence, from Def. \ref{indep model}, it follows that $\mathcal{M}_G \neq \mathcal{M}_H$.
We distinguish between two cases for $H$, namely (a) $v_i \leftarrow v_j \notin H$, and (b) $v_i \leftarrow v_j \in H$.

Case (a): Suppose (i) also holds for $H$; that is, $w_i \nci v_j|w_j \; [\overline{H}]$. 
Then since $v_i \rightarrow v_j \notin H$, it follows from $c$-separation that the variable $v_i$ must be connected to $v_j$ by directed path(s) whose interior vertices must only belong to $N(V) \setminus \{v_i,v_j\}$.
Thus $H$ implies the relation $w_i \ci v_j | w_j , (v_k)_{k \in V \setminus \{i,j\}} \; [\overline{H}]$, so that (ii) cannot also hold.

Case (b): Since $v_i \leftarrow v_j \in H$ it follows from $c$-separation that $w_j \nci v_i | w_i \; [\overline{H}]$, so that (iii) does not hold for $H$.
\end{proof}
\begin{remark}
More generally the same argument shows that a DAG $G \in \mathcal{G}$ satisfies $v_i \rightarrow v_j \notin G$ if and only if $\exists S \subseteq \text{pa}_G(j) \setminus \{i\}$ such that $w_i \ci v_j| w_j, (v_k)_{k \in S} \; [\overline{G}]$.
As a consequence, we have the  interpretation that conditional upon the (secondary variables) $(X_i)_{i \in W}$, the solid arrows in Fig. \ref{illustrate} encode conditional independence relationships among the (primary variables) $(Y_i)_{i \in V}$.
This motivates the name ``conditional DAG''.
\end{remark}

It is well known that conditional independence (and causal) relationships can usefully be described through a qualitative, graphical representation.
However to be able to use a graph for {\it reasoning} it is necessary for that graph to embody certain assertions that themselves obey a logical calculus.
\cite{Pearl4} proposed such a set of rules (the semi-graphoid axioms) that any reasonable set of assertions about how one set of variables might be irrelevant to the prediction of a second, given the values of a third, might hold \citep[see also][]{Dawid3,Studeny}.
This can then be extended to causal assertions \citep{Pearl}, thereby permitting study of causal hypotheses and their consequences without the need to first construct elaborate probability spaces and their extensions under control.
Below we establish that the independence models $\mathcal{M}_G$ induced by $c$-separation on CDAGs are semi-graphoids and thus enable reasoning in the present setting with two kinds of variables:

\begin{lemma}[Semi-graphoid]
For any DAG $G \in \mathcal{G}$, the set $\mathcal{M}_G$ is {\it semi-graphoid} \citep{Pearl2}. That is to say, for all disjoint $A,B,C,D \subseteq N(V) \cup N(W)$ we have
\begin{enumerate}[(i)]
\item (triviality) $A \ci \emptyset | C \; [\mathcal{M}_G]$
\item (symmetry) $A \ci B | C \; [\mathcal{M}_G]$ implies $B \ci A | C \; [\mathcal{M}_G]$
\item (decomposition) $A \ci B,D | C \; [\mathcal{M}_G]$ implies $A \ci D | C \; [\mathcal{M}_G]$
\item (weak union) $A \ci B,D | C \; [\mathcal{M}_G]$ implies $A \ci B | C,D \; [\mathcal{M}_G]$
\item (contraction) $A \ci B | C,D \; [\mathcal{M}_G]$ and $A \ci D | C \; [\mathcal{M}_G]$ implies $A \ci B,D | C \; [\mathcal{M}_G]$
\end{enumerate}
\end{lemma}
\begin{proof}
The simplest proof is to note that our notion of $c$-separation is equivalent to classical $d$-separation applied to an extension $\underline{G}$ of the CDAG $\overline{G}$.
The semi-graphoid properties then follow immediately by the facts that (i) $d$-separation satisfies the semi-graphoid properties \citep[p.48 of][]{Studeny}, and (ii) the restriction of a semi-graphoid to a subset of vertices is itself a semi-graphoid \citep[p.14 of][]{Studeny}.

Construct an extended graph $\underline{G}$ from the CDAG $\overline{G}$ by the addition of a node $z$ and directed edges from $z$ to each of the secondary vertices $N(W)$.
Then for disjoint $A,B,C \subseteq N(V) \cup N(W)$ we have that $A$ and $B$ are $c$-separated by $C$ in $\overline{G}$ if and only if $A$ and $B$ are $d$-separated by $C$ in $\underline{G}$.
This is because every path in the undirected graph $U_4(\overline{G})$ (recall the definition of $c$-separation) that contains an edge $w_i \rightarrow w_j$ corresponds uniquely to a path in $U_3(\underline{G})$ that contains the sub-path $w_i \rightarrow z \rightarrow w_j$.
\end{proof}

\subsection{Causal CDAG models}
\label{section_causal_CDAG}
The previous section defined CDAG models using the framework of formal independence models.
However, CDAGs can also be embellished with a causal interpretation, that we make explicit below.
In this paper we make a causal sufficiency assumption that the $(X_i)_{i \in W}$ are the only source of confounding for the $(Y_i)_{i \in V}$ and below we talk about direct causes at the level of $(X_i)_{i \in W} \cup (Y_i)_{i \in V}$.

\begin{definition}[Causal CDAG]
A CDAG is \emph{causal} when an edge $v_i \rightarrow v_j$ exists if and only if $Y_i$ is a direct cause of $Y_j$. 
It is further assumed that $X_i$ is a direct cause of $Y_i$ and not a direct cause of $Y_j$ for $j \neq i$.
Finally it is assumed that no $Y_i$ is a direct cause of any $X_j$.
\end{definition}
\begin{remark}
Here \emph{direct cause} is understood to mean that the parent variable has a ``controlled direct effect'' on the child variable in the framework of Pearl \citep[e.g. Def. 4.5.1 of][]{Pearl} (it is not necessary that the effect is {\it physically} direct).
No causal assumptions are placed on interaction between the secondary variables $(X_i)_{i \in W}$.
\end{remark}
\begin{remark}
In a causal CDAG the secondary variables $(X_i)_{i \in W}$ share some of the properties of instrumental variables \citep{Didilez}.
Consider estimating the average causal effect of $Y_i$ on $Y_j$. Then, conditioning on $X_j$ in the following, $X_i$ can be used as a natural experiment \citep{Greenland} to determine the size and sign of this causal effect.
When we are interested in the controlled direct effect, we can repeat this argument with additional conditioning on the $(Y_k)_{k \in V \setminus \{i,j\}}$ (or a smaller subset of conditioning variables if the structure of $G$ is known).
\end{remark}

\subsection{Identifiability of CDAGs} \label{theory}

There exist well-known identifiability results for independence models $\mathcal{M}$ that are induced by Bayesian networks; see for example \cite{Spirtes,Pearl}.
These relate the observational distribution $\mathbb{P}^{(Y_i)}$ of the random variables $(Y_i)_{i \in V}$ to an appropriate DAG representation by means of $d$-separation, Markov and faithfulness assumptions (discussed below).
The problem of identification for CDAGs is complicated by the fact that (i) the primary variables $(Y_i)_{i \in V}$ are insufficient for identification, (ii) the joint distribution $\mathbb{P}^{(X_i) \cup (Y_i)}$ of the primary variables $(Y_i)_{i \in V}$ and the secondary variables $(X_i)_{i \in W}$ need not be Markov with respect to the CDAG $\overline{G}$, and (iii) we must work with the alternative notion of $c$-separation.
Below we propose novel ``partial'' Markov and faithfulness conditions that will permit, in the next section, an identifiability theorem for CDAGs.
We make the standard assumption that there is no selection bias (for example by conditioning on common effects).

\begin{Assumption}[Existence]
There exists a true CDAG $\overline{G}$.
In other words, the observational distribution $\mathbb{P}^{(X_i) \cup (Y_i)}$ induces an independence model that can be expressed as $\mathcal{M}_G$ for some DAG $G \in \mathcal{G}$.
\end{Assumption}

\begin{definition}[Partial Markov] \label{markov}
Let $G$ denote the true DAG.
We say that the observational distribution $\mathbb{P}^{(X_i) \cup (Y_i)}$ is \emph{partially Markov} with respect to $G$ when the following holds:
For all disjoint subsets $\{ i \},\{j\},C \subseteq \{1,\dots,p\}$ we have $w_i \ci v_j | w_j, (v_k)_{k \in C} \; [\mathcal{M}_G] \Rightarrow X_i \ci Y_j | X_j, (Y_k)_{k \in C}$.
\end{definition}

\begin{definition}[Partial faithfulness] \label{faithful}
Let $G$ denote the true DAG.
We say that the observational distribution $\mathbb{P}^{(X_i) \cup (Y_i)}$ is \emph{partially faithful} with respect to $G$ when the following holds:
For all disjoint subsets $\{ i \},\{j\},C \subseteq \{1,\dots,p\}$ we have $w_i \ci v_j | w_j, (v_k)_{k \in C} \; [\mathcal{M}_G] \Leftarrow X_i \ci Y_j | X_j, (Y_k)_{k \in C}$
\end{definition}

\begin{remark}
The partial Markov and partial faithfulness properties do not place any constraint on the marginal distribution $\mathbb{P}^{(X_i)}$ of the secondary variables.
\end{remark}

The following is an immediate corollary of Lem. \ref{equiv}:

\begin{theorem}[Identifiability] \label{idenfity}
Suppose that the observational distribution $\mathbb{P}^{(X_i) \cup (Y_i)}$ is partially Markov and partially faithful with respect to the true DAG $G$. Then
\begin{enumerate}[(i)]
\item It is not possible to identify the true DAG $G$ based on the observational distribution $\mathbb{P}^{(Y_i)}$ of the primary variables alone.
\item It is possible to identify the true DAG $G$ based on the observational distribution $\mathbb{P}^{(X_i) \cup (Y_i)}$.
\end{enumerate}
\end{theorem}
\begin{proof}
(i) We have already seen that $\mathbb{P}^{(Y_i)}$ is not Markov with respect to the DAG $G$:
Indeed a statistical association $Y_i \nci Y_j| (Y_k)_{k \in V \setminus \{i,j\}}$ observed in the distribution $\mathbb{P}^{(Y_i)}$ could either be due to a direct interaction $Y_i \rightarrow Y_j$ (or $Y_j \rightarrow Y_i$), or could be mediated entirely through variation in the secondary variables $(X_k)_{k \in W}$.
(ii) It follows immediately from Lemma \ref{equiv} that observation of both the primary and secondary variables $(Y_i)_{i \in V} \cup (X_i)_{i \in W}$ is sufficient to facilitate the identification of $G$.
\end{proof}

\subsection{Estimating CDAGs from data}

In this section we assume that the partial Markov and partial faithfulness properties hold, so that the true DAG $G$ is identifiable from the joint observational distribution of the primary and secondary variables.
Below we consider score-based estimation for CDAGs and prove consistency of certain score-based CDAG estimators.

\begin{definition}[Score function; \cite{Chickering}]
A \emph{score function} is a map $S:\mathcal{G} \rightarrow [0,\infty)$ with the interpretation that if two DAGs $G,H \in \mathcal{G}$ satisfy $S(G)<S(H)$ then $H$ is preferred to $G$.
\end{definition}

We will study the asymptotic behaviour of $\hat{G}_S$, the estimate of graph structure obtained by maximising $S(G)$ over all $G \in \mathcal{G}$ based on observations $(X_i^j,Y_i^j)_{i=1,\dots,p}^{j=1,\dots,n}$.
Let $\mathbb{P}_n = \mathbb{P}^{(X_i^j,Y_i^j)}$ denote the finite-dimensional distribution of the $n$ observations.

\begin{definition}[Partial local consistency]
We say the score function $S$ is \emph{partially locally consistent} if, whenever $H$ is constructed from $G$ by the addition of one edge $Y_i \rightarrow Y_j$, we have
\begin{enumerate}
\item $X_i \nci Y_j | X_j, (Y_k)_{k \in \text{pa}_G(j)} \Rightarrow \lim_{n \rightarrow \infty}\mathbb{P}_n[S(H) > S(G)] = 1$
\item $X_i \ci Y_j | X_j, (Y_k)_{k \in \text{pa}_G(j)} \Rightarrow \lim_{n \rightarrow \infty}\mathbb{P}_n[S(H) < S(G)] = 1$.
\end{enumerate}
\end{definition}

\begin{theorem}[Consistency]
If $S$ is partially locally consistent then $\lim_{n \rightarrow \infty} \mathbb{P}_n[\hat{G}_S = G] = 1$, so that $\hat{G}_S$ is a consistent estimator of the true DAG $G$.
\end{theorem}
\begin{proof}
It suffices to show that $\lim_{n \rightarrow \infty} \mathbb{P}_n[\hat{G}_S = H] = 0$ whenever $H \neq G$.
There are two cases to consider:

Case (a): Suppose $v_i \rightarrow v_j \in H$ but $v_i \rightarrow v_j \notin G$.
Let $H'$ be obtained from $H$ by the removal of $v_i \rightarrow v_j$.
From $c$-separation we have $w_i \ci v_j | w_j, (v_k)_{k \in \text{pa}_G(j)} \; [\overline{G}]$ and hence from the partial Markov property we have $ X_i \ci Y_j | X_j, (Y_k)_{k \in \text{pa}_G(j)}$. 
Therefore if $S$ is partially locally consistent then $\lim_{n \rightarrow \infty}\mathbb{P}_n[S(H) < S(H')] = 1$, so that $\lim_{n \rightarrow \infty} \mathbb{P}_n[\hat{G}_S = H] = 0$.

Case (b): Suppose $v_i \rightarrow v_j \notin H$ but $v_i \rightarrow v_j \in G$.
Let $H'$ be obtained from $H$ by the addition of $v_i \rightarrow v_j$.
From $c$-separation we have $w_i \nci v_j | w_j, (v_k)_{k \in \text{pa}_G(j)} \; [\overline{G}]$ and hence from the partial faithfulness property we have $ X_i \nci Y_j | X_j, (Y_k)_{k \in \text{pa}_G(j)}$.
Therefore if $S$ is partially locally consistent then $\lim_{n \rightarrow \infty}\mathbb{P}_n[S(H) < S(H')] = 1$, so that $\lim_{n \rightarrow \infty} \mathbb{P}_n[\hat{G}_S = H] = 0$.
\end{proof}

\begin{remark}
In this paper we adopt a {\it maximum a posteriori} (MAP) -Bayesian approach and consider score functions given by a posterior probability $p(G|(x_i^l,y_i^l)_{i = 1,\dots,p}^{l=1,\dots,n})$ of the DAG $G$ given the data $(x_i^l,y_i^l)_{i=1,\dots,p}^{l=1,\dots,n}$. 
This requires that a prior $p(G)$ is specified over the space $\mathcal{G}$ of DAGs.
From the partial Markov property we have that, for $n$ independent observations, such score functions factorise as 
\begin{eqnarray}
S(G) = p(G) \prod_{l=1}^n p((x_i^l)_{i = 1,\dots,p}) \prod_{i = 1}^p p(y_i^l|(y_k^l)_{k \in \text{pa}_G(i)},x_i^l). \label{factorise}
\end{eqnarray}
We further assume that the DAG prior $p(G)$ factorises over parent sets $\text{pa}_G(i) \subseteq V \setminus \{i\}$ as 
\begin{eqnarray}
p(G) = \prod_{i = 1}^p p(\text{pa}_G(i)).
\end{eqnarray}
This implies that the score function in Eqn. \ref{factorise} is \emph{decomposable} and the maximiser $\hat{G}_S$, i.e. the MAP estimate, can be obtained via integer linear programming. In the supplement we derive an integer linear program that targets the CDAG $\hat{G}_S$ and thereby allows exact (i.e. deterministic) estimation in this class of models.
\end{remark}

\begin{lemma} \label{character}
A score function of the form Eqn. \ref{factorise} is partially locally consistent if and only if, whenever $H$ is constructed from $G$ by the addition of one edge $v_i \rightarrow v_j$, we have
\begin{enumerate}
\item $X_i \nci Y_j | X_j, (Y_k)_{k \in \text{pa}_G(j)} \Rightarrow \lim_{n \rightarrow \infty} \mathbb{P}_n[B_{H,G} > 1] = 1$
\item $X_i \ci Y_j | X_j, (Y_k)_{k \in \text{pa}_G(j)} \Rightarrow \lim_{n \rightarrow \infty} \mathbb{P}_n[B_{H,G} < 1] = 1$
\end{enumerate}
where 
\begin{eqnarray}
B_{H,G} = \frac{ p((Y_j^l)^{l = 1,\dots,n}|(Y_k^l)_{k \in \text{pa}_H(j)}^{l = 1,\dots,n},(X_j^l)^{l = 1,\dots,n}) }{ p((Y_j^l)^{l = 1,\dots,n}|(Y_k^l)_{k \in \text{pa}_G(j)}^{l = 1,\dots,n},(X_j^l)^{l = 1,\dots,n}) }
\end{eqnarray}
is the Bayes factor between two competing local models $\text{pa}_G(j)$ and $\text{pa}_H(j)$.
\end{lemma}

\subsection{Bayes factors and common variables}

The characterisation in Lemma \ref{character} justifies the use of any consistent Bayesian variable selection procedure to obtain a score function.
The secondary variables $(X_i^l)^{l = 1,\dots,n}$ are included in all models, and parameters relating to these variables should therefore share a common prior.
Below we discuss a formulation of the Bayesian linear model that is suitable for CDAGs.

Consider variable selection for node $j$ and candidate parent (index) set $\text{pa}_G(j) = \pi \subseteq V \setminus \{j\}$.
We construct a linear model for the observations
\begin{eqnarray}
Y_j^l = [1 \; X_j^l] \bm{\beta}_0 + \bm{Y}_{\pi}^l \bm{\beta}_\pi + \epsilon_j^l, \; \; \; \epsilon_j^l \sim N(0,\sigma^2)
\end{eqnarray}
where $\bm{Y}_\pi^l = (Y_k^l)_{k \in \pi}$ is used to denote a row vector and the noise $\epsilon_j^l$ is assumed independent for $j = 1,\dots,p$ and $l = 1,\dots,n$.
Although suppressed in the notation, the parameters $\bm{\beta}_0$, $\bm{\beta}_\pi$ and $\sigma$ are specific to node $j$.
This regression model can be written in vectorised form as
\begin{eqnarray}
\bm{Y}_j =  \bm{M}_0 \bm{\beta}_0 +  \bm{Y}_\pi \bm{\beta}_\pi + \bm{\epsilon}
\end{eqnarray}
where $\bm{M}_0$ is the $n \times 2$ matrix whose rows are the $[1 \; X_j^l]$ for $l = 1,\dots,n$ and $\bm{Y}_\pi$ is the matrix whose rows are $\bm{Y}_\pi^l$ for $l = 1,\dots,n$.

We orthogonalize the regression problem by defining $\bm{M}_\pi = (\bm{I} - \bm{M}_0 (\bm{M}_0^T\bm{M}_0)^{-1}\bm{M}_0^T) \bm{Y}_{\pi}$ so that the model can be written as
\begin{eqnarray}
\bm{Y}_j =  \bm{M}_0 \tilde{\bm{\beta}}_0 +  \bm{M}_\pi \tilde{\bm{\beta}}_\pi + \bm{\epsilon}
\end{eqnarray}
where $\tilde{\bm{\beta}}_0$ and $\tilde{\bm{\beta}}_\pi$ are Fisher orthogonal parameters \citep[see][for details]{Deltell}.

In the conventional approach of \cite{Jeffreys}, the prior distribution is taken as
\begin{eqnarray}
p_{j,\pi}(\tilde{\bm{\beta}}_\pi,\tilde{\bm{\beta}}_0,\sigma) & = & p_{j,\pi}(\tilde{\bm{\beta}}_\pi|\tilde{\bm{\beta}}_0,\sigma) p_j(\tilde{\bm{\beta}}_0,\sigma) \\
p_j(\tilde{\bm{\beta}_0},\sigma) & \propto & \sigma^{-1} \label{reference prior}
\end{eqnarray}
where Eqn. \ref{reference prior} is the reference or independent Jeffreys prior.
(For simplicity of exposition we leave conditioning upon $\bm{M}_0$, and $\bm{M}_\pi$ implicit.)
The use of the reference prior here is motivated by the  observation that the common parameters $\bm{\beta}_0,\sigma$ have the same meaning in each model $\pi$ for variable $Y_j$ and should therefore share a common prior distribution \citep{Jeffreys}. Alternatively, the prior can be motivated by invariance arguments that derive $p(\bm{\beta}_0,\sigma)$ as a right Haar measure \citep{Bayarri}.
Note however that $\sigma$ does not carry the same meaning across $j \in V$ in the application that we consider below, so that the prior is specific to fixed $j$.
For the parameter prior $p_{j,\pi}(\tilde{\bm{\beta}}_\pi|\tilde{\bm{\beta}}_0,\sigma)$ we use the $g$-prior \citep{Zellner}
\begin{eqnarray} 
\tilde{\bm{\beta}}_\pi | \tilde{\bm{\beta}}_0,j,\pi \sim N(\bm{0},g \sigma^2 (\bm{M}_\pi^T \bm{M}_\pi)^{-1}) \label{gprior}
\end{eqnarray}
where $g$ is a positive constant to be specified.
Due to orthogonalisation, $\text{cov}(\hat{\bm{\beta}}_\pi) = \sigma^2 (\bm{M}_\pi^T\bm{M}_\pi)^{-1}$ where $\hat{\bm{\beta}}_\pi$ is the maximum likelihood estimator for $\tilde{\bm{\beta}}_\pi$, so that the prior is specified on the correct length scale \citep{Deltell}. 
We note that many alternatives to Eqn. \ref{gprior} are available in the literature \citep[including][]{Johnson,Bayarri}.

Under the prior specification above, the marginal likelihood for a candidate model $\pi$ has the following closed-form expression:
\begin{eqnarray}
p_j(\bm{y}_j|\pi) & = & \frac{1}{2} \Gamma\left(\frac{n-2}{2}\right) \frac{1}{\pi^{(n-2)/2}} \frac{1}{|\bm{M}_0^T \bm{M}_0|^{1/2}} \left(\frac{1}{g+1}\right)^{|\pi|/2}  b^{-(n-2)/2} \\
b & = & \bm{y}_j^T \left( \bm{I} - \bm{M}_0(\bm{M}_0^T \bm{M}_0)^{-1} \bm{M}_0^T - \frac{g}{g+1}\bm{M}_\pi (\bm{M}_\pi^T \bm{M}_\pi)^{-1} \bm{M}_\pi^T \right) \bm{y}_j 
\end{eqnarray}

For inference of graphical models, multiplicity correction is required to adjust for the fact that the size of the space $\mathcal{G}$ grows super-exponentially with the number $p$ of vertices \citep{Consonni}.
Following \cite{Scott}, we control multiplicity via the prior
\begin{eqnarray}
p(\pi) \propto \binom{p}{|\pi|}^{-1}.
\end{eqnarray}

\begin{theorem}[Consistency]
Let $g = n$.
Then the Bayesian score function $S(G)$ defined above is partially locally consistent, and hence the corresponding estimator $\hat{G}_S$ is consistent.
\end{theorem}
\begin{proof}
This result is an immediate consequence of Lemma \ref{character} and the well-known variable selection consistency property for the unit-information $g$-prior \citep[see e.g.][]{Fernandez}.
\end{proof}

\section{Results} \label{results}

\begin{table}
\resizebox{\textwidth}{!}{
\begin{tabular}{lccccccccc}
\hline
$\theta = 0$ & \multicolumn{3}{c}{$n=10$} & \multicolumn{3}{c}{$n=100$} & \multicolumn{3}{c}{$n=1000$} \\ \cmidrule(lr){2-4} \cmidrule(lr){5-7} \cmidrule(lr){8-10}
& DAG & DAG2 & CDAG & DAG & DAG2 & CDAG & DAG & DAG2 & CDAG \\
$p=5$ & 2.9 $\pm$ 0.48 & 2.6 $\pm$ 0.48 & 3.4 $\pm$ 0.56 & 3.2 $\pm$ 0.81 & 1.9 $\pm$ 0.43 & 0.8 $\pm$ 0.25 & 3 $\pm$ 0.76 & 1.6 $\pm$ 0.48 & 0.3 $\pm$ 0.21 \\
$p=10$ & 9.8 $\pm$ 0.55 & 9.2 $\pm$ 0.66 & 8.8 $\pm$ 0.81 & 8.2 $\pm$ 0.99 & 5 $\pm$ 0.84 & 2.8 $\pm$ 0.51 & 5.5 $\pm$ 1.2 & 5.4 $\pm$ 1.1 & 0.3 $\pm$ 0.15 \\
$p=15$ & 15 $\pm$ 1.3 & 14 $\pm$ 1.1 & 15 $\pm$ 1.2 & 11 $\pm$ 1 & 6.8 $\pm$ 0.83 & 4.4 $\pm$ 0.58 & 6.3 $\pm$ 1.5 & 8.2 $\pm$ 0.92 & 0.8 $\pm$ 0.25 \\
\hline \\[-2ex]
\hline
$\theta = 0.5$ & \multicolumn{3}{c}{$n=10$} & \multicolumn{3}{c}{$n=100$} & \multicolumn{3}{c}{$n=1000$} \\ \cmidrule(lr){2-4} \cmidrule(lr){5-7} \cmidrule(lr){8-10}
& DAG & DAG2 & CDAG & DAG & DAG2 & CDAG & DAG & DAG2 & CDAG \\
$p=5$ & 5.7 $\pm$ 0.56 & 4.5 $\pm$ 0.48 & 4 $\pm$ 0.49 & 3.9 $\pm$ 0.75 & 2.1 $\pm$ 0.62 & 0.6 $\pm$ 0.31 & 3.8 $\pm$ 0.74 & 1.8 $\pm$ 0.81 & 0 $\pm$ 0 \\
$p=10$ & 9.8 $\pm$ 0.96 & 8.1 $\pm$ 0.95 & 7.8 $\pm$ 1.1 & 7.2 $\pm$ 1.7 & 4.6 $\pm$ 1.2 & 1.7 $\pm$ 0.84 & 7.9 $\pm$ 1.3 & 3 $\pm$ 0.52 & 0.6 $\pm$ 0.31 \\
$p=15$ & 14 $\pm$ 0.85 & 13 $\pm$ 0.86 & 13 $\pm$ 0.7 & 14 $\pm$ 1.1 & 6.2 $\pm$ 0.55 & 3.6 $\pm$ 0.76 & 11 $\pm$ 0.93 & 7.5 $\pm$ 1.7 & 1 $\pm$ 0.45 \\
\hline \\[-2ex]
\hline
$\theta = 0.99$ & \multicolumn{3}{c}{$n=10$} & \multicolumn{3}{c}{$n=100$} & \multicolumn{3}{c}{$n=1000$} \\ \cmidrule(lr){2-4} \cmidrule(lr){5-7} \cmidrule(lr){8-10}
& DAG & DAG2 & CDAG & DAG & DAG2 & CDAG & DAG & DAG2 & CDAG \\
$p=5$ & 5.1 $\pm$ 0.72 & 4 $\pm$ 0.56 & 4.2 $\pm$ 0.49 & 4.7 $\pm$ 0.79 & 2.1 $\pm$ 0.5 & 0.5 $\pm$ 0.31 & 2.7 $\pm$ 0.42 & 1.7 $\pm$ 0.56 & 0.9 $\pm$ 0.48 \\
$p=10$ & 9.1 $\pm$ 0.84 & 8.1 $\pm$ 0.72 & 7.8 $\pm$ 1.3 & 9.7 $\pm$ 1.1 & 3.7 $\pm$ 0.52 & 2.5 $\pm$ 0.45 & 9.4 $\pm$ 1 & 4.4 $\pm$ 0.64 & 0.3 $\pm$ 0.3 \\
$p=15$ & 15 $\pm$ 1.1 & 13 $\pm$ 0.94 & 13 $\pm$ 1.2 & 13 $\pm$ 1.4 & 6.6 $\pm$ 0.91 & 4.7 $\pm$ 0.86 & 17 $\pm$ 1.8 & 9.1 $\pm$ 0.69 & 0.7 $\pm$ 0.4 \\
\hline\hline
\end{tabular}
}
\caption{Simulated data results. 
Here we display the mean structural Hamming distance from the estimated to the true graph structure, computed over 10 independent realisations, along with corresponding standard errors. 
[Data were generated using linear-Gaussian structural equations. $\theta \in [0,1]$ captures the amount of dependence between the secondary variables $(X_i)_{i \in W}$, $n$ is the number of data points and $p$ is the number of primary variables $(Y_i)_{i \in V}$. ``DAG'' = estimation based only on primary variables $(Y_i)_{i \in V}$, ``DAG2'' = estimation based on the full data $(X_i)_{i \in W} \cup (Y_i)_{i \in V}$, ``CDAG'' = estimation based on the full data and enforcing CDAG structure.]
}
\label{tab res 2}
\end{table}

\subsection{Simulated data} \label{sim data}
We simulated data from linear-Gaussian structural equation models (SEMs). Here we summarise the simulation procedure, with full details provided in the supplement.
We first sampled a DAG $G$ for the primary variables and a second DAG $G'$ for the secondary variables (independently of $G$), following a sampling procedure described in the supplement.  That is, $G$ is the causal structure of interest, while $G'$ governs dependence between the secondary variables.
Data for the  secondary variables $(X_i)_{i \in W}$ were  generated from an SEM with structure $G'$.
The strength of dependence between secondary variables was controlled by a parameter $\theta \in [0,1]$.
Here $\theta=0$ renders the secondary variables independent and $\theta=1$ corresponds to a deterministic relationship between secondary variables, with intermediate values of $\theta$ giving different degrees of covariation among the secondary variables.
Finally, conditional on the $(X_i)_{i \in W}$, we simulated data for the primary variables $(Y_i)_{i \in V}$ from an SEM with structure $G$.
To manage computational burden, for all estimators we considered only  models of size $|\pi| \leq 5$.
Performance was quantified by the structural Hamming distance (SHD) between the estimated DAG $\hat{G}_S$ and true, data-generating DAG $G$; we report the mean SHD as computed over 10 independent realisations of the data.

\begin{figure}[t!]
\centering
\includegraphics[width = 0.6\textwidth,clip,trim = 3.5cm 8.75cm 4.25cm 9cm]{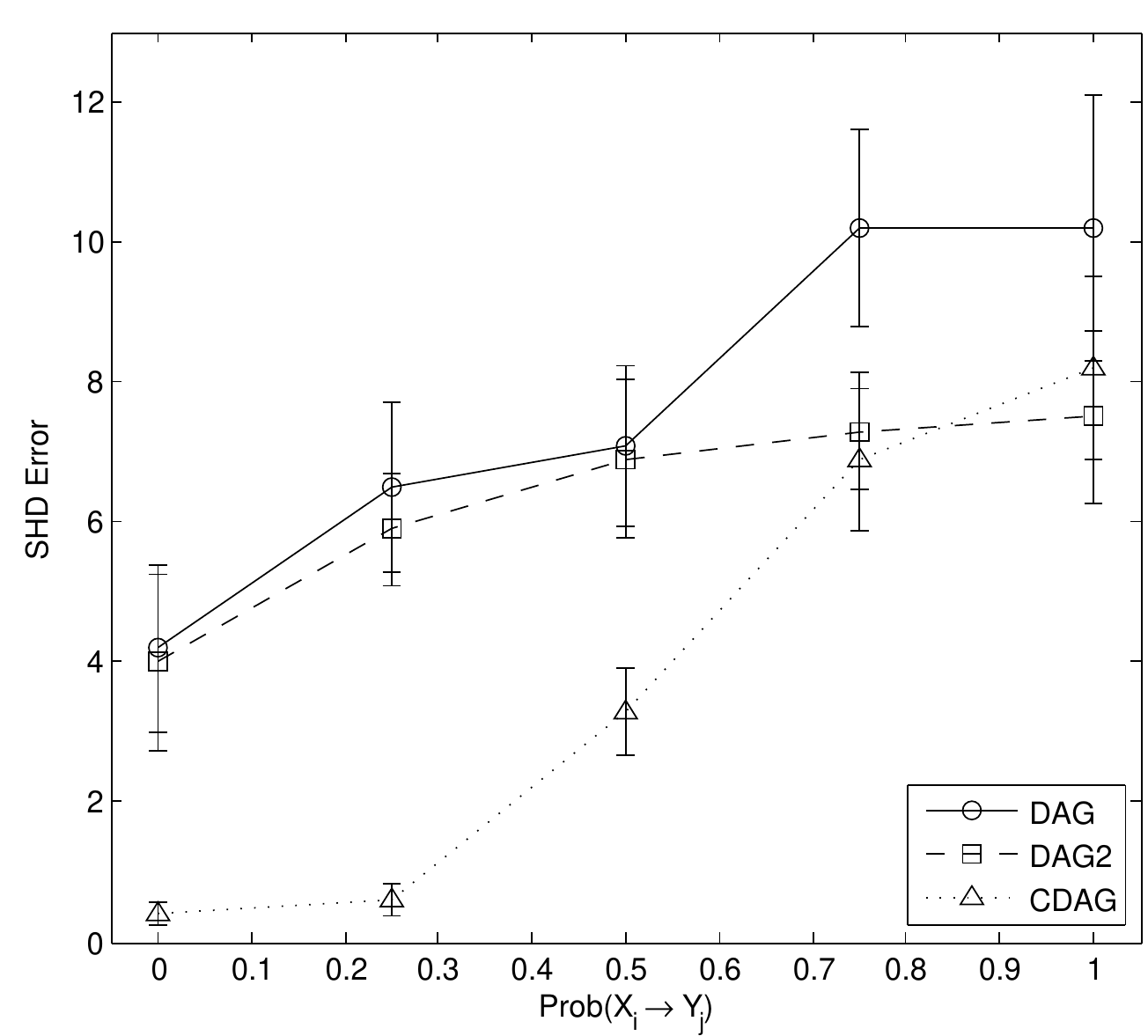}
\caption{Simulated data results; model misspecification.
[Data were generated using linear-Gaussian structural equations. Here we fixed $\theta = 0$, $p = 15$, $n = 1000$ and considered varying the number $E$ of misspecified edges as described in the main text.
On the $x$-axis we display the marginal probability that any given edge $Y_i \rightarrow Y_j$ has an associated misspecified edge $X_i \rightarrow Y_j$, so that when $\text{Prob}(X_i \rightarrow Y_j) = 1$ the number $E$ of misspecified edges is equal to the number of edges in $G$.
``DAG'' = estimation based only on primary variables $(Y_i)_{i \in V}$, ``DAG2'' = estimation based on the full data $(X_i)_{i \in W} \cup (Y_i)_{i \in V}$, ``CDAG'' = CDAG estimation based on the full data $(X_i)_{i \in W} \cup (Y_i)_{i \in V}$.]}
\label{misspec}
\end{figure}

In Table \ref{tab res 2} we compare the proposed score-based CDAG estimator with the corresponding score-based DAG estimator that uses only the primary variable data $(Y_i^l)_{i = 1,\dots,p}^{l=1,\dots,n}$.
We also considered an alternative (DAG2) where a standard DAG estimator is applied to {\it all} of the variables $(X_i)_{i \in W}$, $(Y_i)_{i \in V}$, with the subgraph induced on the primary variables giving the estimate for $G$.
We considered values $\theta = 0$, $0.5$, $0.99$ corresponding to zero, mild and strong covariation among the secondary variables.
In each data-generating regime we found that CDAGs were either competitive with, or (more typically) more effective than, the DAG and DAG2 estimators.
In the $p=15$, $n=1000$ regime (that is closest to the proteomic application that we present below) the CDAG estimator dramatically outperforms these two alternatives.
Inference using DAG2 and CDAG (that both use primary as well as secondary variables) is robust to large values of $\theta$, whereas in the $p=15$, $n=1000$ regime, the performance of the DAG estimator based on only the primary variables deteriorates for large $\theta$.
This agrees with intuition since association between the secondary $X_i$'s may induce correlations between the primary $Y_i$'s. 
CDAGs showed superior efficiency at large $n$ compared with DAG2.
This reflects the unconstrained nature of the DAG2 estimator that must explore a wider class of conditional independence models, resulting in a loss of precision relative to CDAGs that was evident in Table \ref{tab res 2}.

To better understand the limitations of CDAGs we considered a data-generating regime that violated the CDAG assumptions. 
We focused on the $\theta = 0$, $p=15$, $n=1000$ regime where the CDAG estimator performs well when data are generated ``from the model''.
We then introduced a number $E$ of edges of the form $X_i \rightarrow Y_j$ where $Y_i \rightarrow Y_j \in G$. These edges (strongly) violate the structural assumptions implied by the CDAG model because their presence means that $X_i$ is no longer a suitable instrument for $Y_i \rightarrow Y_j$ as it is no longer conditionally independent of the variable $Y_j$ given $Y_i$.
We assessed performance of the CDAG, DAG and DAG2 estimators as the number $E$ of such misspecified edges is increased (Fig. \ref{misspec}) We find that whilst CDAG continues to perform well up to a moderate fraction of misspecified edges,  for larger fractions performance degrades and eventually coincides with DAG and DAG2.

\subsection{TCGA patient data} \label{real data}

In this section we illustrate  the use of CDAGs in an analysis of proteomic data from cancer samples. 
We focus on causal links between post-translationally modified proteins involved in a process called cell signalling. The estimation of causal signalling networks has been a prominent topic in computational biology for some years \citep[see, among others,][]{Sachs,Nelander,Hill,Oates}. Aberrations to causal signalling networks are central to cancer biology \citep{Weinberg}.

In this application, the primary variables $(Y_i)_{i \in V}$ represent abundance of phosphorylated protein (p-protein) while the secondary variables $(X_i)_{i \in W}$ represent abundance of corresponding total proteins (t-protein). A  t-protein can be modified by a process called phosphorylation to form the corresponding p-protein and the p-proteins play a key role in signalling. 
An edge $v_i \rightarrow v_j$ has the biochemical interpretation that the phosphorylated form of protein $i$ acts as a causal influence on phosphorylation of protein $j$. 
The data we analyse are from the TCGA ``pan-cancer" project \citep{Akbani} and comprise measurements of protein levels (including both t- and p-proteins) using a technology called reverse phase protein arrays (RPPAs). We focus on $p=24$ proteins for which (total, phosphorylated) pairs are available; the data span eight different cancer types \citep[as defined in][]{Stadler} with a total sample size of $n=3,467$ patients.
 We first illustrate the key idea of using secondary variables to inform causal inference regarding primary variables with an  example from the TCGA data:

\begin{figure}[t!]
\centering
\begin{subfigure}[c]{0.45\textwidth}
\centering
\vspace{20pt}
\begin{tabular}{|l|c|c|} \hline
protein & node & variable \\ \hline
CHK1 & $w_i$ & $X_i$ \\
p-CHK1 & $v_i$ & $Y_i$ \\
CHK2 & $w_j$ & $X_j$ \\
p-CHK2 & $v_j$ & $Y_j$ \\ \hline
\end{tabular}
\vspace{5pt}
\caption{}
\label{key}
\end{subfigure} 
\begin{subfigure}[]{0.45\textwidth}
\centering
\includegraphics[width = 0.5\textwidth,clip,trim = 2.2cm 24cm 15cm 2.5cm]{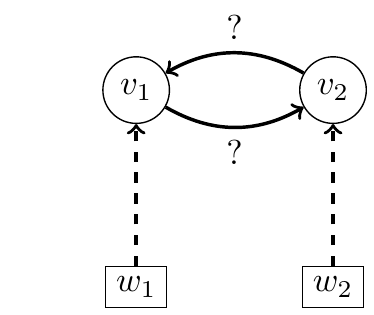}
\caption{}
\label{IVs}
\end{subfigure} 
\begin{subfigure}[c]{\textwidth}
\includegraphics[width = 1.03\textwidth,clip,trim = 0cm 5.5cm 0cm 0cm]{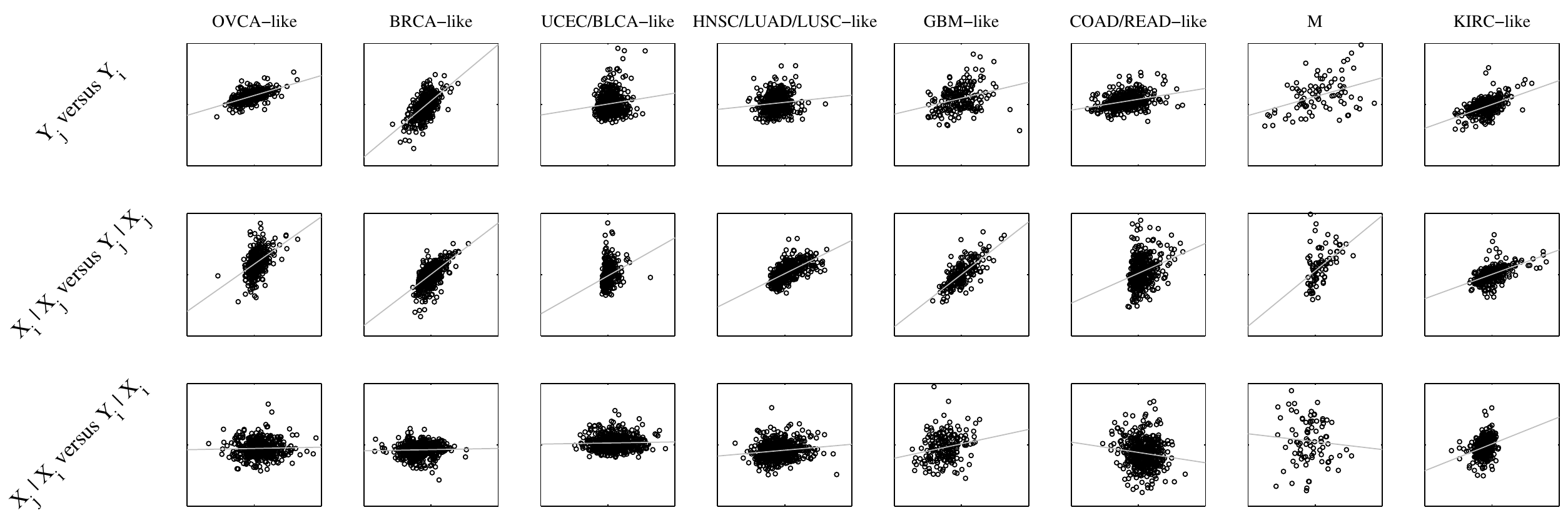}
\caption{}
\label{example}
\end{subfigure} 
\caption{CHK1 total protein (t-CHK1) as a natural experiment for phosphorylation of CHK2 (p-CHK2).
(a) Description of the variables.
(b) A portion of the CDAG relating to these variables. It is desired to estimate whether there is a causal relationship $Y_i \rightarrow Y_j$ (possibly mediated by other protein species) or {\it vice versa}.
(c) Top row: Plotting phosphorylated CHK1 (p-CHK1; $Y_i$) against p-CHK2 ($Y_j$) we observe weak correlation in some of the cancer subtypes. 
Middle row: We plot the residuals when t-CHK1 is regressed on total CHK2 (t-CHK2; x-axis) against the residuals when p-CHK2 is regressed on t-CHK2 (y-axis). The plots show a strong (partial) correlation in each subtype that suggests a causal effect in the direction p-CHK1 $\rightarrow$ p-CHK2.
Bottom row: Reproducing the above but with the roles of CHK1 and CHK2 reversed, we see much reduced and in many cases negligible partial correlation, suggesting lack of a causal effect in the reverse direction, i.e. p-CHK1 $\nleftarrow$ p-CHK2.
[The grey line in each panel is a least-squares linear regression.]
}
\end{figure}

\begin{example}[CHK1 t-protein as a natural experiment for CHK2 phosphorylation] \label{example1}
Consider RVs $(Y_i,Y_j)$ corresponding respectively to p-CHK1 and p-CHK2, the phosphorylated forms of CHK1 and CHK2 proteins.
Fig. \ref{example} (top row) shows that these variables are weakly correlated in most of the 8 cancer subtypes.
There is a strong partial correlation between t-CHK1 ($X_i$) and p-CHK2 ($Y_j$) in each of the  subtypes when conditioning on t-CHK2 ($X_j$) (middle row), but there is essentially no partial correlation between t-CHK2 ($X_j$) and p-CHK1 ($Y_i$) in the same subtypes when conditioning on  t-CHK1 (bottom row).
{\it Thus, under the CDAG assumptions, this suggests that there exists a directed causal path from p-CHK1 to p-CHK2, but not vice versa.}
\end{example}

\begin{figure}[t!]
\includegraphics[width = \textwidth]{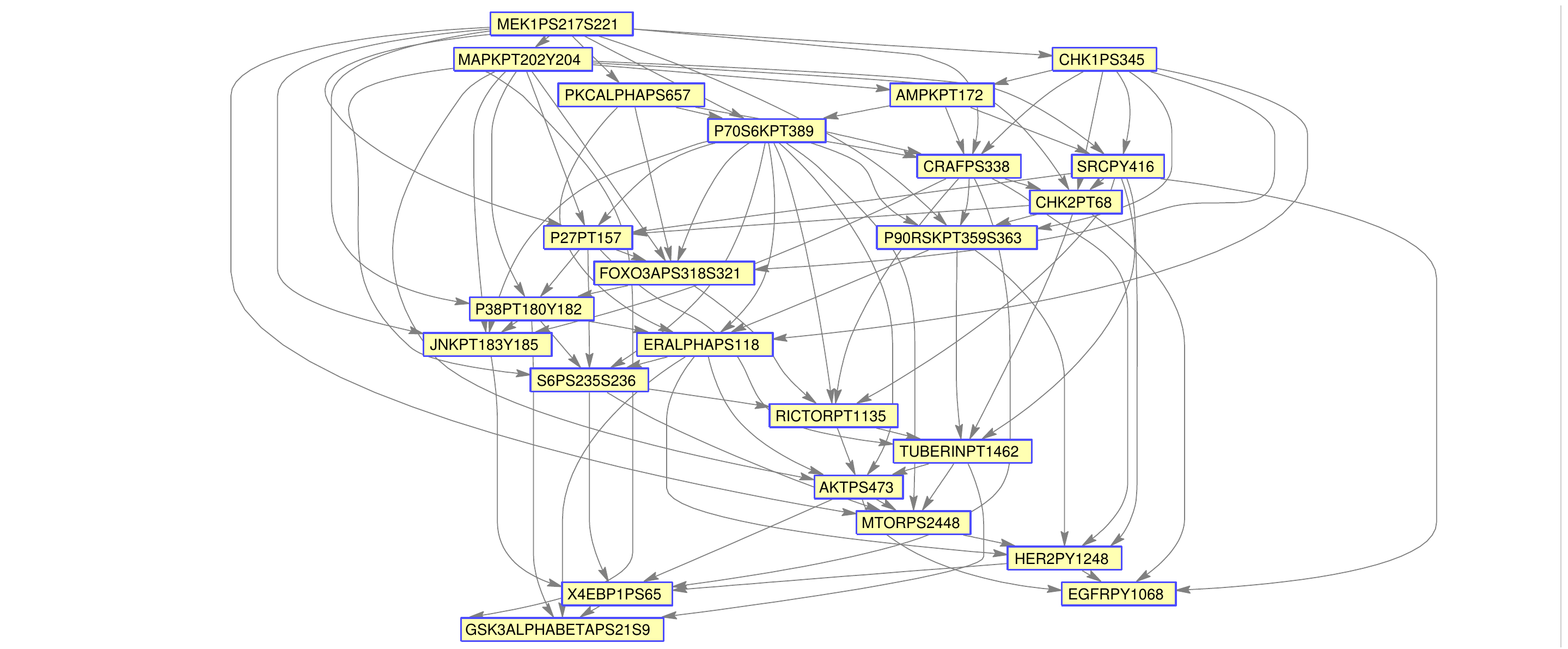}
\caption{{\it Maximum a posteriori} conditional DAG, estimated from proteomic  data derived from cancer patient samples 
(from the Cancer Genome Atlas pan-cancer study, samples belonging to the BRCA-like group as defined by \citet{Stadler}).
Here vertices represent phosphorylated proteins (primary variables) and edges have the biochemical  interpretation that the parent protein plays a causal role in phosphorylation of the child protein.}
\label{brca cdag}
\end{figure}

Example \ref{example1} provides an example from the TCGA data where controlling for a secondary variable (here, t-protein abundance) may be important for causal inference concerning primary variables (p-protein abundance).

We now apply the CDAG methodology to all $p=24$ primary variables that we consider, using data from the largest subtype in the study, namely BRCA-like (see \citet{Stadler} for details concerning the data and subtypes). The estimated graph is shown in Fig. \ref{brca cdag}. 
We note that assessment of the biological validity of this causal graph is a nontrivial matter, and outside the scope of the present paper.
However, we observe that several well known edges, such as from p-MEK to p-MAPK, appear in the estimated graph and are oriented in the expected direction. Interestingly, in several of these cases, the edge orientation is different when a standard DAG estimator is applied to the same data, showing that the CDAG formulation can reverse edge orientation with respect to a classical DAG (see supplement). 
We note also that the CDAG is  denser, with more edges, than the  DAG (Fig. \ref{density}), demonstrating that in many cases, accounting for secondary variables can render candidate edges more salient.
These differences support our theoretical results insofar as they demonstrate that  in practice CDAG estimation can give quite different results from a DAG analysis of the same primary variables but we note that proper assessment of estimated causal structure in this setting requires further biological work that is beyond the scope of this paper.

\section{Conclusions}

Practitioners of causal inference understand that it is important to distinguish between variables that could reasonably be considered as potential causes and those that cannot.
In this work we put forward CDAGs as a simple class of graphical models
that make this distinction explicit.
Motivated by molecular biological applications, we developed CDAGs that use bijections between primary and secondary index sets. However, the general approach presented here could be extended to other multivariate settings where variables are in some sense non-exchangeable.
Naturally many of the philosophical considerations 
and practical limitations and caveats
of  classical DAGs remain relevant for CDAGs and we refer the reader to \cite{Dawid} for an illuminating discussion of these issues.

The application to proteomic  data presented above represents a principled  approach to integrate total and phosphorylated protein data for causal inference. Our results suggest that in some settings it may be important to account for total protein levels in analysing protein phosphorylation and CDAGs allow such integration in a causal framework.
Theoretical and empirical results showed  that CDAGs can  improve estimation of causal structure relative to classical DAGs when the CDAG assumptions are even approximately satisfied.

\begin{figure}[t!]
\centering
\includegraphics[width = 0.6\textwidth]{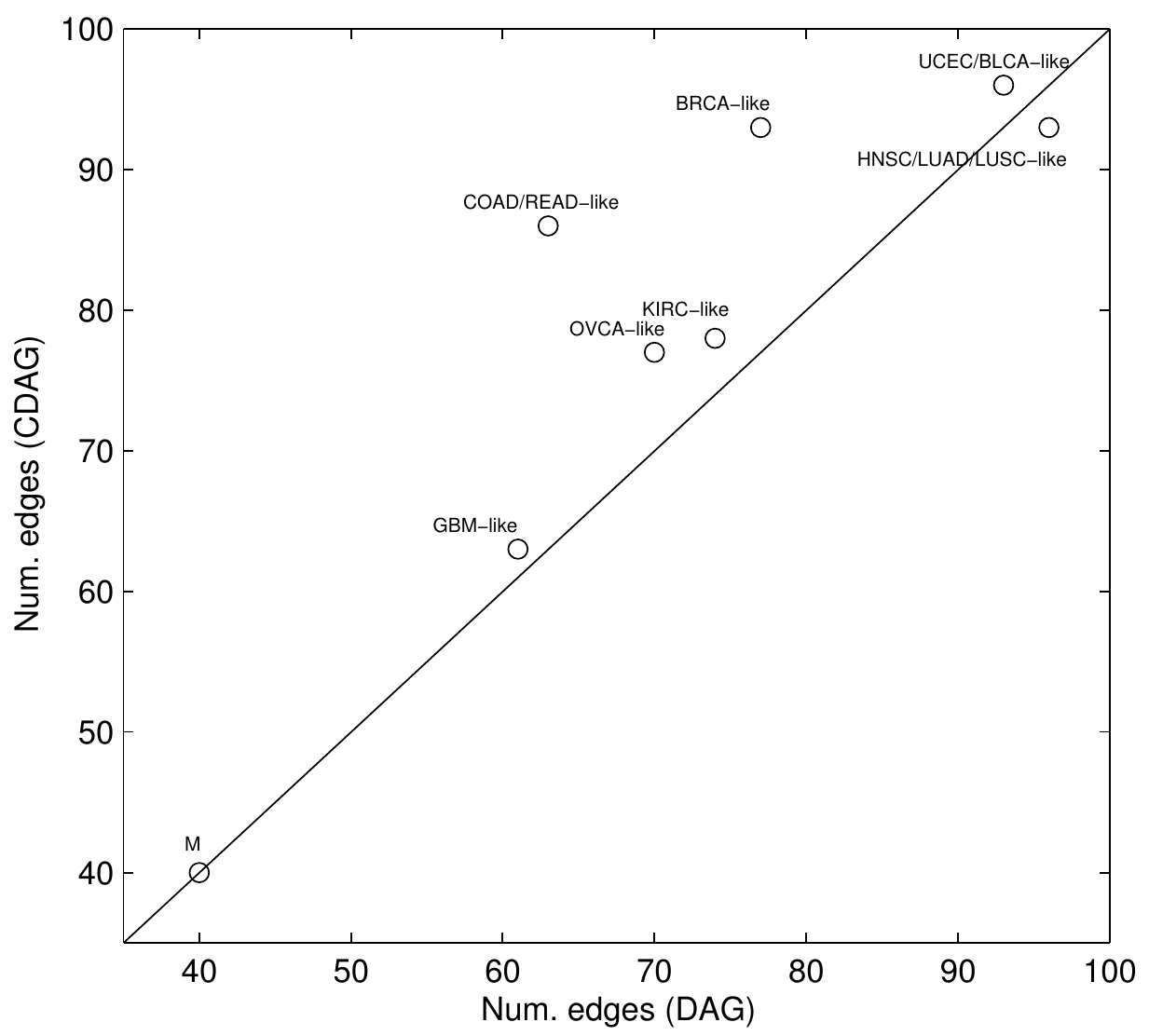}
\caption{Cancer patient data; relative density of estimated protein networks. 
Here we plot the number of edges in the networks inferred by estimators based on DAGs ($x$-axis) and based on CDAGs ($y$-axis).
[Each point corresponds to a cancer subtype that was previously defined by St\''{a}dler {\it et al}, 2014.
``DAG'' = estimation based only on primary variables $(Y_i)_{i \in V}$, ``CDAG'' = estimation based on the full data and enforcing CDAG structure.]}
\label{density}
\end{figure}

We briefly mention three natural extensions of the present work:
(i) The CDAGs put forward here allow exactly one secondary variable $X_i$ for each primary variable $Y_i$.
In many settings this may be overly restrictive. To return to  the protein example, it may be useful to consider multiple phosphorylation sites for a single total protein, i.e. multiple $Y_i$ for a single $X_i$; this would be a natural extension of the CDAG methodology. Examples of this more general formulation were recently discussed by \cite{Neto} in the context of eQTL data.
Conversely we could extend the recent ideas of \cite{Kang} by allowing for multiple secondary variables for each primary variable, not all of which may be valid as instruments.
(ii) In many applications data may be available from multiple related but causally non-identical groups, for example  disease types.
It could then be useful to consider {\it joint} estimation of multiple CDAGs, following recent work on estimation for multiple DAGs  \citep{Oates7,Oates8}.
(iii) Advances in assay-based technologies now mean that the number $p$ of variables is frequently very large. Estimation for high-dimensional CDAGs may be possible using  recent results for high-dimensional DAGs  \citep[e.g.][and others]{Kalisch,Loh}.

\subsubsection*{Acknowledgments}
This work was inspired by discussions with Jonas Peters at the Workshop \emph{Statistics for Complex Networks}, Eindhoven 2013, and Vanessa Didelez at the \emph{UK Causal Inference Meeting}, Cambridge 2014.
CJO was supported by the Centre for Research in Statistical Methodology (CRiSM) EPSRC EP/D002060/1. SM acknowledges the support of the UK Medical Research Council and is a recipient of a Royal Society Wolfson Research Merit Award.

\end{document}